\documentclass[a4paper]{llncs}
\usepackage{color}
\usepackage{amsmath}
\usepackage{amssymb}
\usepackage{cleveref}
\usepackage{tikz}
\usepackage{pgfplots}
\usepackage{multirow}

\crefrangeformat{equation}{(#3#1#4--#5#2#6)}
\crefname{equation}{}{}

\newcommand{\ve}[1]{\boldsymbol{#1}}

\title{Bounds on Codes Correcting Tandem and Palindromic Duplications}
\author{Andreas Lenz \inst{1} \thanks{This work was supported by the Institute for Advanced Study (IAS), Technische Universit\"{a}t M\"{u}nchen (TUM), with funds from the German Excellence Initiative and the European Union's Seventh Framework Program (FP7) under grant agreement no.~291763.} \and Antonia Wachter-Zeh \inst{1} $^\star$ \and Eitan Yaakobi \inst{2} }
\institute{Institute for Communications Engineering, \\Technical University of Munich, Germany \and Computer Science Department, \\ Technion -- Israel Institute of Technology, Haifa, Israel \\
	\email{andreas.lenz@mytum.de, antonia.wachter-zeh@tum.de, yaakobi@cs.technion.ac.il}}

\begin{document}
	
	\maketitle
	\begin{abstract}
		In this work, we derive upper bounds on the cardinality of tandem duplication and palindromic deletion correcting codes by deriving the generalized sphere packing bound for these error types. We first prove that an upper bound for tandem \emph{deletions} is also an upper bound for \emph{inserting} the respective type of duplications. Therefore, we derive the bounds based on these special \emph{deletions} as this results in tighter bounds. We determine the spheres for tandem and palindromic duplications/deletions and the number of words with a specific sphere size. 
		Our upper bounds on the cardinality directly imply lower bounds on the redundancy which we compare with the redundancy of the best known construction correcting arbitrary burst errors.
		Our results indicate that the correction of palindromic duplications requires more redundancy than the correction of tandem duplications. Further, there is a significant gap between the minimum redundancy of duplication correcting codes and burst insertion correcting codes.
	\end{abstract}

	\section{Introduction}
	The increasing demand for high density and long-term data storage and the recent advance in biotechnological methodology has motivated the storage of digital data in DNA. One interesting application in this area involves the storage of data in the DNA of living organisms. Tagging genetically modified organisms, infectious bacteria, conducting biogenetical studies or storing data are only a few in a list of modern applications. However, the data is corrupted by errors during the replication of DNA and therefore an adequate error protection mechanism has to be found. Typical errors include point insertions, deletions, substitutions and tandem or palindromic duplications. While the correction of substitutions, insertions and deletions is well studied, knowledge about correcting tandem and palindromic duplication errors is relatively limited. In the former case, a subsequence of the original word is duplicated and inserted directly after the original subsequence.
	An example for a tandem duplication of length~$3$ in a DNA sequence $GATCATG$ is $GATC\underline{ATC}ATG$, where the underlined part highlights the duplication. Similarly, a palindromic duplication in the same word is $GATC\underline{CTA}ATG$. The focus of this paper is to determine the minimum redundancy needed to correct tandem and palindromic duplications.
	
	The idea of upper bounding the sizes of insertion/deletion correcting codes by the fractional transversal number of the associated hypergraph has been introduced in \cite{Kulkarni13}.
	In \cite{Fazeli15}, this procedure has been analyzed and generalized to other error models, such as the $Z$-channel, grain-error channel, and projective spaces. Further, it has been shown that the average sphere packing value provides a valid upper bound on code sizes, if the associated hypergraph is regular and symmetric. Repetition errors form a related error model to tandem and palindromic duplications and corresponding error correcting codes have been studied in \cite{Dolecek10}. Codes correcting tandem duplications have been considered in \cite{Jain16}, where amongst others a construction for the correction of an arbitrary number of fixed length duplications was presented. These codes are based on choosing irreducible words with respect to tandem duplications and their relation to zero run-length-limited systems has been illustrated. In this work we employ the method presented in \cite{Fazeli15}, known as the generalized sphere packing bound for tandem and palindromic duplications.
	\subsection{Preliminaries} \label{ss:preliminaries}
	We denote $\ve{x} = (x_1, x_2, \dots, x_n) \in \mathbb{Z}_q^n$ to be a vector of $n$ symbols over the ring of integers modulo $q$, $x_i \in \mathbb{Z}_q \, \forall \, i$. The length of a word $\boldsymbol{x}$ is denoted by $|\ve{x}|$. A tandem duplication of length $\ell$ at position~$i$ with $0\leq i \leq n-\ell$ in a word $\ve{x} = (\ve{uvw})$, with $|\ve{u}| = i, |\ve{v}| = \ell, |\ve{w}| = n-\ell-i$ is defined by $\tau_{i,\ell}(\ve{x}) = (\ve{uvvw}) \in \mathbb{Z}_q^{n+\ell}$ and a palindromic duplication is defined by $\pi_{i,\ell}(\ve{x}) = (\ve{uv} \ve{v}^\mathrm{R} \ve{w})$, where $\ve{v}^\mathrm{R} = (v_\ell v_{\ell-1} \dots v_1)$ is the reversal of $\ve{v}$. The inverse operation, a tandem deletion of length $\ell$ at position $0\leq i \leq n-2\ell$ in a word $\ve{x} = (\ve{uvvw})$ with $|\ve{u}| = i, |\ve{v}| = \ell, |\ve{w}| = n-2\ell-i$ is denoted by $\tau^\delta_{i,\ell}(\ve{x}) = (\ve{uvw}) \in \mathbb{Z}_q^{n-\ell}$. Finally, we write a palindromic deletion of length $\ell$ at position $0\leq i \leq n-2\ell$ in a word $\ve{x} = (\ve{uv}\ve{v}^\mathrm{R}\ve{w})$ with $|\ve{u}| = i, |\ve{v}| = \ell, |\ve{w}| = n-2\ell-i$ as $\pi^\delta_{i,\ell}(\ve{x}) = (\ve{uvw}) \in \mathbb{Z}_q^{n-\ell}$. Note that the deletion operations are only defined at positions $i$, where the word $\ve{x}$ is of the form $(\ve{uvvw})$, respectively $(\ve{uv}\ve{v}^\mathrm{R}\ve{w})$ with $|\ve{u}| = i$. With these definitions, the sphere of a word $\ve{x}$ is the set of all vectors that are reached by a maximum of $t$ tandem or palindromic duplications, respectively deletions, i.e.,
	\begin{equation} \label{eq:spheres}
	S_{\epsilon, \ell, t}(\ve{x}) = \{\ve{y} | \ve{y} = \epsilon_{i_1, \ell}(...(\epsilon_{i_\theta, \ell}(\ve{x}))...), \theta \leq t\},
	\end{equation}
	with $\epsilon \in  \{\tau, \pi, \tau^\delta, \pi^\delta \} $. Here, $\tau$ and $\tau^\delta$ denote tandem duplications, respectively deletions and $\pi$, $\pi^\delta$ denote palindromic duplications and deletions. By this definition $\ve{x} \in S_{\epsilon, \ell, t}(\ve{x})$ and the size of these sets depends on $\ve{x}$, which is the key complication when computing upper bounds on the code cardinality. For a word~$\ve{x}$, let $r(\ve{x})$ be the number of \emph{runs}, respectively $r_i(\ve{x})$ the number of runs of length~$i$ and $r_{\geq i}(\ve{x})$ be the number of runs of length at least $i$ in $\ve{x}$. For example, the sequence $\ve{x} = (ATTTAAC)$ has $4$ runs, where $2$ are of length $1$, one is of length $2$ and one of length $3$. A codebook $\mathcal{C} \subset \mathbb{Z}_q^n$ is called $t$-\emph{tandem duplication} (\emph{palindromic duplication}, \emph{tandem deletion}, \emph{palindromic deletion}) correcting, if $S_{\epsilon, \ell, t}(\ve{c}_1) \cap S_{\epsilon, \ell, t}(\ve{c}_2) \neq \emptyset$ implies $\ve{c}_1 = \ve{c}_2$ for all $\ve{c}_1, \ve{c}_2 \in \mathcal{C}$. In the following we will omit the index $t$ for the case $t=1$ for the above definitions and use the term \emph{single}-error correcting.
	\section{Relationship of Duplication and Deletion Codes}
	\label{sec:relationship_duplication_deletion_codes}
	\subsection{Equivalence of Tandem Duplication and Deletion Codes} \label{ss:equivalence_tandem_duplication_deletion_codes}
	For insertion and deletion correcting codes, it has been shown that a code $\mathcal{C}$ is $t$-insertion correcting if and only if it is $t$-deletion correcting \cite{Levenshtein66}. A similar behavior can be shown for tandem duplications. To do so, we start with some terminology for tandem duplications, that has been introduced in \cite{Jain16}.
	\begin{definition}[$\ell$-step derivative]
		From $\ve{x} \in \mathbb{Z}_q^n$ we define the $\ell$-step derivative $\phi_\ell(\ve{x})=(\ve{u}, \ve{v})$ with $\ve{u} = (x_1, x_2, \dots, x_\ell)$ and $\ve{v} = (x_{\ell+1}, x_{\ell+2}, \dots, x_n) - $ $(x_1, x_2, \dots, x_{n-\ell})$.
	\end{definition}
	It has been shown in \cite{Jain16} that a tandem duplication of length $\ell$ in $\ve{x}$ corresponds to an insertion of $\ell$ consecutive zeros in $\ve{v}$. This motivates the introduction of the \emph{trunk} and \emph{zero-signature} representation for $\ve{v}$. 
	\begin{definition}[Trunk and $\ell$-zero signature]
		Let $0^m$ denote the $m$-fold repetition of $0$ and let $\ve{v} = (0^{m_0},w_1,0^{m_1},w_2, \dots, w_{p}, 0^{m_p})$ with $w_i \in \mathbb{Z}_q \setminus \{0\}$ and $p = wt_{\mathrm{H}}(\ve{v})$ be the Hamming weight of $\ve{v}$. We then define the \emph{trunk} of $\ve{v}$ to be $\mu_\ell(\ve{v}) = (0^{m_0 \bmod \ell},w_1,0^{m_1 \bmod \ell},w_2, \dots, w_{p}, 0^{m_p \bmod \ell})$ as the word that is obtained by shortening every zeros run of length $m$ to be of length $m \bmod \ell$. Further, the $\ell$-\emph{zero signature} of $\ve{v}$ is defined as $\sigma_\ell(\ve{v}) = \left(\left\lfloor\frac{m_0}{\ell}\right\rfloor, \left\lfloor\frac{m_1}{\ell}\right\rfloor, \dots, \left\lfloor\frac{m_p}{\ell}\right\rfloor\right)$.
	\end{definition}
	Note that $\ve{v}$ is uniquely determined by its trunk $\mu_\ell(\ve{v})$ and zero-signature $\sigma_\ell(\ve{v})$. Further, let $\rho_\ell(\ve{x})$ be the \emph{root} of $\ve{x}$, which is defined as the tuple $\rho_\ell(\ve{x}) = (\ve{u}, \mu_\ell(\ve{v}))$, where $(\ve{u}, \ve{v})$ is the $\ell$-step derivative of $\ve{x}$. Note that for notational reasons this definition of the root is slightly different from that in \cite{Jain16}. It is easy to see that a tandem duplication in $\ve{x}$ corresponds to increasing an entry of $\sigma_\ell(\ve{v})$ by 1 (a tandem deletion corresponds to decreasing the entry by 1), but leaves the root $\rho_\ell(\ve{x})$ unchanged. With these definitions it is possible to introduce the following metric, that is closely related to tandem duplications.
	\begin{definition}[Tandem duplication distance]
		For two words $\ve{x}_1, \ve{x}_2 \in \mathbb{Z}_q^n$ with $\ell$-step derivatives $\phi_\ell(\ve{x}_1) = (\ve{u}_1, \ve{v}_1)$ and $\phi_\ell(\ve{x}_2) = (\ve{u}_2, \ve{v}_2)$ we define the \emph{tandem duplication distance} to be
		\begin{equation}
		d_{\tau, \ell}(\ve{x}_1, \ve{x}_2) = \left\{ \begin{array}{ll}
		\infty, & \mathrm{if} \, \rho_\ell(\ve{x}_1) \neq \rho_\ell(\ve{x}_2) \\
		|\sigma_\ell(\ve{v}_1) - \sigma_\ell(\ve{v}_2)|_1, \quad & \mathrm{if} \, \rho_\ell(\ve{x}_1) = \rho_\ell(\ve{x}_2)
		\end{array} \right.,
		\end{equation}
		with the $l_1$-norm $|\sigma_\ell(\ve{v})|_1 = \sum_{j=1}^{s}|\sigma_\ell(\ve{v})_j|$, where $s=| \sigma_\ell(\ve{v}) |$.
	\end{definition}
	From the results in \cite{Jain16}, we can directly deduce the following two corollaries.
	\begin{corollary} \label{cor:tan_dup_distance}
		A code $\mathcal{C} \subset \mathbb{Z}_q^n$ is $t$-tandem duplication correcting (length $\ell$) if and only if $d_{\tau, \ell}(\ve{c}_1, \ve{c}_2) \geq 2t+1 \,\forall\, \ve{c}_1, \ve{c}_2 \in \mathcal{C}, \ve{c}_1 \neq \ve{c}_2$.
	\end{corollary}
	\begin{corollary} \label{cor:tan_del_distance}
		A code $\mathcal{C} \subset \mathbb{Z}_q^n$ is $t$-tandem deletion correcting (length $\ell$) if and only if $d_{\tau, \ell}(\ve{c}_1, \ve{c}_2) \geq 2t+1 \,\forall\, \ve{c}_1, \ve{c}_2 \in \mathcal{C}, \ve{c}_1 \neq \ve{c}_2$.
	\end{corollary}
	Note that Corollary \ref{cor:tan_dup_distance} and \ref{cor:tan_del_distance} imply that if for a codeword $\ve{c}_1 \in \mathcal{C}$ of a $t$-tandem duplication correcting code $\mathcal{C}$ there exists another codeword $\ve{c}_2 \in \mathcal{C}$ with $\rho_\ell(\ve{c}_1) = \rho_\ell(\ve{c}_2)$, we have $|\sigma_\ell(\ve{c}_1)|_1 = |\sigma_\ell(\ve{c}_2)|_1 > t$. Let us formulate the central theorem of this section which follows from Corollaries \ref{cor:tan_dup_distance} and \ref{cor:tan_del_distance}.
	\begin{theorem} \label{thm:equivalence_tandem_duplication_deletion}
		A code $\mathcal{C} \subset \mathbb{Z}_q^n$ is $t$-tandem duplication correcting if and only if it is $t$-tandem deletion correcting.
	\end{theorem}
	\subsection{Relationship between Palindromic Duplication and Deletion codes}
	For palindromic duplication errors, an equivalence similar to Theorem \ref{thm:equivalence_tandem_duplication_deletion} does not hold. A counter example for $\ell=2, t=1$ that shows that not every palindromic \emph{deletion} correcting code is palindromic \emph{duplication} correcting is presented here.
	\begin{example}
		Let $\mathcal{C} = \{\ve{c}_1, \ve{c}_2\}$ with $\ve{c}_1 = (010101)$ and $\ve{c}_2 = (010011)$. $\mathcal{C}$ is single palindromic \emph{deletion} correcting, since $B^{\rho_2^D}_1(\ve{c}_1) = \{\ve{c}_1\}$ and $B^{\rho_2^D}_1(\ve{c}_2) = \{\ve{c}_2, (0101)\}$ and thus $B^{\rho_2^D}_1(\ve{c}_1) \cap B^{\rho_2^D}_1(\ve{c}_2)= \emptyset$. On the other hand, $\mathcal{C}$ is not single palindromic \emph{duplication} correcting since $B^{\rho_2}_1(\ve{c}_1) \cap B^{\rho_2}_1(\ve{c}_2) = \{(01001101)\}$.
	\end{example}
	The following example illustrates that also not every palindromic \emph{duplication} correcting code is palindromic \emph{deletion} correcting.
	\begin{example}
		Consider the code $\mathcal{C} = \{\ve{c}_1, \ve{c}_2\}$ with $\ve{c}_1 = (011010)$ and $\ve{c}_2 = (011110)$. $\mathcal{C}$ is single palindromic \emph{duplication} correcting, since 
		\begin{align*}
		B^{\rho_2}_1(\ve{c}_1) &= \{\ve{c}_1, (01101010), (01111010), (01100110), (01101100), (01101001)\}, \\ B^{\rho_2}_1(\ve{c}_2) &= \{\ve{c}_2, (01101110), (01111110), (01111001)\},
		\end{align*}
		and thus $B^{\rho_2}_1(\ve{c}_1) \cap B^{\rho_2}_1(\ve{c}_2)= \emptyset$. However, $\mathcal{C}$ is not single palindromic \emph{deletion} correcting since $B^{\rho_2^D}_1(\ve{c}_1) \cap B^{\rho_2^D}_1(\ve{c}_2) = \{(0110)\}$.
	\end{example}
	\section{Sphere Sizes for Tandem and Palindromic Duplications and Deletions} \label{sec:error_spheres_tandem_palindromic_dupliations_deletions}
	In the following we derive the size of the spheres $S_{\epsilon, \ell, t}(\ve{x})$, see \eqref{eq:spheres}, for tandem and palindromic duplication and deletion errors. Note that by the definition of the error sphere \eqref{eq:spheres}, $S_{\epsilon, \ell, t}(\ve{x})$ contains $\ve{x}$, which results in a sphere size that is equal to the number of descendants plus one. For the subsequent two lemmas we denote $\phi_\ell(\ve{x})=(\ve{u}, \ve{v})$, according to the definition from Section \ref{ss:equivalence_tandem_duplication_deletion_codes}.
	\subsection{Tandem Duplication Sphere}
	\begin{lemma}
		\begin{equation*}
		|S_{\tau, \ell, t}(\ve{x})| = \sum_{j=0}^t \binom{wt_{\mathrm{H}}(\ve{v})+j}{j} = \binom{t+wt_{\mathrm{H}}(\ve{v})+1}{t}.
		\end{equation*}
	\end{lemma}
	\begin{proof}
		Recall that a tandem duplication error corresponds to increasing one entry of the zero signature $\sigma_\ell(\ve{v})$ by one. Then, the duplication sphere size equals the number of zero signatures $\ve{\sigma} \in \mathbb{N}_0^{wt_{\mathrm{H}}(\ve{v})+1}$ with $\sigma_i \geq \sigma_\ell(\ve{v})_i$ and $|\ve{\sigma}|_1-|\sigma_\ell(\ve{v})|_1 \leq t$. \qed
	\end{proof}
	\subsection{Tandem Deletion Sphere}
	\begin{lemma}\label{lemma:tandem_deletion_sphere}
		\begin{equation*}
		|S_{\tau^\delta, \ell, 1}(\ve{x})| = wt_{\mathrm{H}}(\sigma_\ell(\ve{v})) + 1.
		\end{equation*}
	\end{lemma}
	\begin{proof}
		It is only possible to delete a tandem duplication at positions, where $\sigma_\ell(\ve{v})_j > 0$. Further, we add one as $S_{\tau^\delta, \ell, 1}(\ve{x})$ contains $\ve{x}$. \qed
	\end{proof}

	\subsection{Palindromic Duplication Sphere}
	The size of the palindromic duplication sphere is not straightforward to derive due to the knotted nature of \cref{eq:cond_j_less_l_a,eq:cond_j_less_l_b,eq:cond_j_less_l_c} and \cref{eq:cond_j_geq_l_a,eq:cond_j_geq_l_b,eq:cond_j_geq_l_c}. We start with deriving the palindromic duplication sphere size for the cases $\ell=1$ and $\ell=2$. For $\ell=1$, a palindromic duplication is a single duplication. Therefore, the sphere size is
	\begin{equation}
	|S_{\pi,1}(\ve{x})| = r(\ve{x})+1,
	\end{equation}
	as duplications in the same run yield the same outcome.
	\begin{lemma}\label{lemma:palindromic_duplication_sphere_l_2}The size of the palindromic duplication sphere $|S_{\pi, 2}(\ve{x})|$ is
		\begin{equation}
		|S_{\pi, 2}(\ve{x})| = n - \sum_{i=3}^n (i-2) r_i(\ve{x}) = 2 r(\boldsymbol{x}) - r_1(\boldsymbol{x}).
		\end{equation}
	\end{lemma}
	\begin{proof}
		We start with the observation that there are $n-1$ possible positions $i \in \{0,1, ..., n-2\}$ for palindromic duplications. Now, for $\ell=2$, the conditions $\pi_{i,\ell}(\ve{x}) = \pi_{i+j,\ell}(\ve{x})$ \cref{eq:cond_j_less_l_a,eq:cond_j_less_l_b,eq:cond_j_less_l_c} and \cref{eq:cond_j_geq_l_a,eq:cond_j_geq_l_b,eq:cond_j_geq_l_c} become  $x_1 = x_2 = \dots = x_{2+j} \, \forall \, j>0$. We therefore deduce that two palindromic duplications in $\ve{x}$ of length $2$ only result in the same vector $\boldsymbol{y}=\pi_{i,\ell}(\ve{x}) = \pi_{i+j,\ell}(\ve{x})$ iff they appear in the same run in $\ve{x}$. Further, two palindromic duplications at two different positions $i$ and $i+j, j>0$ can only duplicate symbols from the same run, if this run has length at least $3$. Thus, every additional symbol to runs of length at least $2$ does not increase the duplication sphere size and has to be subtracted from the palindromic duplication sphere size. Using $\sum_{i=1}^n i r_i(\ve{x}) = n$ and $\sum_{i=1}^n r_i(\ve{x}) = r(\ve{x})$ yields the statement. \qed
	\end{proof}
	For $\ell \geq 3$ and $j \geq 2$, \cref{eq:cond_j_less_l_a,eq:cond_j_less_l_b,eq:cond_j_less_l_c} and \cref{eq:cond_j_geq_l_a,eq:cond_j_geq_l_b,eq:cond_j_geq_l_c} do not imply $x_1=x_2 = \dots =x_{\ell+j}$. For example, consider $\ell=3$ and the word $\ve{x} = (ACAACA)$. Then, $\pi_{0,3}(\ve{x}) = \pi_{3,3}(\ve{x}) = (ACAACAACA)$. However, it is possible to find an upper bound on the size of the palindromic duplication sphere. For $j = 1$, \cref{eq:cond_j_less_l_a,eq:cond_j_less_l_b,eq:cond_j_less_l_c} become $x_1 = x_2 = \dots = x_{\ell+1}$. Therefore two neighboring palindromic duplications can only result in the same word if they appear in one run.
	\begin{lemma}
		\begin{equation*}
		|S_{\pi, \ell}(\ve{x})| \leq n - \ell + 2 - \sum_{i=\ell+1}^{n} (i-\ell)r_i(\ve{x})
		\end{equation*}
	\end{lemma}
	\begin{proof}
		There are $n-\ell+1$ possible positions for palindromic duplications of length $\ell$. Now, as seen before, duplications in the same run result in the same descendant. We therefore subtract the additional $i-\ell$ entries of runs with length at least $\ell+1$ from the number of possible positions for duplications to obtain an upper bound on the duplication sphere.  \qed
	\end{proof}
	\subsection{Palindromic Deletion Sphere}
	Similar to the previous section, we start with deriving the size of the palindromic deletions spheres for $\ell=1$ and $\ell=2$. For $\ell=1$, a palindromic deletion is a de-duplication of one symbol. Therefore, the size of the error sphere becomes
	\begin{equation}
	|S_{\pi^\delta, 1}(\ve{x})| = r_{\geq 2}(\ve{x}) + 1,
	\end{equation}
	where $r_{\geq 2}(\ve{x})$ is the number of runs of length at least $2$. Further, we derive the following lemma for binary words.
	\begin{lemma} \label{lemma:palindromic_deletion_sphere_l_2}
		The size of the palindromic deletion sphere $|S_{\pi^\delta, 2}(\ve{x})|$ for $q=2$ is
		\begin{equation*}
		|S_{\pi^\delta, 2}(\ve{x})| = r_{\mathcal{I},2}(\ve{x}) + r_{\geq 4}(\ve{x}) + 1,
		\end{equation*}
		where $r_{\mathcal{I},2}(\ve{x})$ is the number of runs of length $2$, that are located at the interior of $\ve{x}$, i.e., between $x_2$ and $x_{n-1}$. Further, $r_{\geq 4}(\ve{x})$ denotes the number of runs of length at least $4$ in $\ve{x}$.
	\end{lemma}
	\begin{proof}
		There are $4$ possible patterns $(0000)$, $(1111)$, $(0110)$, $(1001)$, at which palindromic deletions of length $2$ can occur. Recall that, as we have seen in the proof of Lemma \ref{lemma:palindromic_duplication_sphere_l_2}, two palindromic deletions of length $2$ at two distinct positions in a word $\ve{x}$ can only results in the same outcome, if they appear in the same run. Every run of length at least $4$ contains one of the patterns $(0000)$, $(1111)$ and therefore will contribute one element to the palindromic deletion sphere. The patterns $(0110)$, $(1001)$ contain a run of length exactly $2$, that is located in the interior of $\ve{x}$, such that there is at least one symbol to the left and right of the run. Thus, every run of length $2$, that is located in the interior of $\ve{x}$ also contributes one unique element in the palindromic deletion sphere. Therefore, counting also the element $\ve{x}$, which is contained in $|S_{\pi^\delta, 2}(\ve{x})|$, the total size of the deletion sphere is $r_{\mathcal{I},2}(\ve{x}) + r_{\geq4}(\ve{x})+1$.  \qed
		
	\end{proof}
	Let us define the matrix $\ve{A}^\pi_\ell(\ve{x}) \in \mathbb{Z}_q^{\ell\times n-2\ell+1}$ to be
	\begin{equation}
	\ve{A}^\pi_\ell(\ve{x}) = \begin{bmatrix}
	x_{2\ell}-x_1 & x_{2\ell+1}-x_2 & \dots & x_n-x_{n-2\ell+1} \\
	x_{2\ell-1}-x_2 & x_{2\ell}-x_3 & \dots & x_{n-1}-x_{n-2\ell} \\
	\vdots & \vdots &\ddots & \vdots \\
	x_{\ell+1}-x_\ell & x_{\ell+2}-x_{\ell+1} & \dots & x_{n-\ell+1}-x_{n-\ell}
	\end{bmatrix}.
	\end{equation}
	With this definition it is directly possible to establish the following upper bound on the size of the palindromic deletion spheres for arbitrary deletion length $\ell$.
	\begin{lemma} \label{lemma:upper_bound_deletion_sphere}
		The palindromic deletion sphere $|S_{\pi^\delta, \ell}(\ve{x})|$ is upper bounded by
		\begin{equation*}
		|S_{\pi^\delta, \ell}(\ve{x})| \leq r^{(0)} \left(\ve{A}^\pi_\ell(\ve{x})\right)+1,
		\end{equation*}
		where $r^{(0)} \left(\ve{A}^\pi_\ell(\ve{x})\right)$ is the number of runs of all zero columns in $\left(\ve{A}^\pi_\ell(\ve{x})\right)$.
	\end{lemma}
	\begin{proof}
		Clearly, a palindrome of length $\ell$ in the word $\ve{x}$ corresponds to a zero column in the matrix $\ve{A}_\ell^\pi(\ve{x})$. Therefore palindromic deletions are only possible at positions $i$, where $\ve{A}_\ell^\pi(\ve{x})$ has a zero-column. Further, it can be shown that two neighboring zero columns are only possible if $x_{i+1} = x_{i+2} = \dots = x_{i+2\ell+1}$, i.e. for a run of length $2\ell+1$. However, two palindromic deletions inside the same run result in the same words. Therefore, every run of all zero columns in $\left(\ve{A}^\pi_\ell(\ve{x})\right)$ contributes one unique element to $S_{\pi^\delta, \ell}(\ve{x})$. \qed
	\end{proof}
	\section{Bounds on the Code Size} \label{sec:code_size_bounds}
	In this section, we derive bounds for tandem and palindromic \emph{deletion} correcting codes. The bound for tandem deletion correcting codes also provides a bound for the duplication correcting codes, since every tandem duplication correcting code is a deletion correcting code, as we have shown in Section \ref{sec:relationship_duplication_deletion_codes}. Note that the bound for palindromic deletion correcting codes is a first step towards understanding palindromic duplication errors. For single error correcting codes, we deduce the following theorem from \cite{Fazeli15}.
	\begin{theorem}
		The maximum cardinality $|C^*(n,\ell)|$ of tandem duplication (palindromic deletion) correcting codes of length $n$ satisfies \cite{Fazeli15}
		\begin{equation} \label{eq:gspb}
		|C^*(n,\ell)| \leq \sum_{\ve{x} \in \mathbb{Z}_q^n \cup \mathbb{Z}_q^{n-\ell}} \frac{1}{|S_{\epsilon, \ell}(\ve{x})|},
		\end{equation}
		where $\epsilon=\tau^\delta$ for tandem duplication and $\epsilon=\pi^\delta$ for palindromic deletion correcting codes.
	\end{theorem}
	Note that it is necessary to formulate the fractional transversal sum over all words of length $n$ and $n-\ell$, since the spheres $S_{\epsilon, \ell}(\ve{x}), \ve{x} \in \mathbb{Z}_q^n$ contain words both of length $n$ and $n-\ell$. The bound \eqref{eq:gspb} can be rephrased to
	\begin{equation} \label{eq:gspb2}
	|C^*(n,\ell)| \leq \sum_{i=0}^{i_{\max}} \frac{N_\epsilon(n,\ell,i)+N_\epsilon(n-\ell,\ell,i)}{i+1},
	\end{equation}
	where $N_{\epsilon}(n,\ell,i) = |\{\ve{x} \in \mathbb{Z}_2^n : |S_{\epsilon, \ell}(\ve{x})| = i+1 \}|$ is the number of words of length $n$ with sphere size $i+1$ and $i_{\max}$ is the maximum sphere size. The next sections are directed towards finding $N_{\epsilon}(n,\ell,i)$ for our error models.
	\subsection{Bound for Tandem Deletions}
	Similar to the strategy in \cite{Fazeli15}, we have to compute the number of words of length $n$ with sphere size $i+1$.
	\begin{lemma}
		\begin{align*}
		N_{\tau^\delta}(n,\ell,i) &= |\{\ve{x} \in \mathbb{Z}_q^n : |S_{\tau^\delta, \ell}(\ve{x})| = i+1 \}| = \\
		&= \sum_{\nu=i}^{\left\lfloor\frac{n}{\ell}\right\rfloor-1}  \sum_{\omega=i-1}^{n-(\nu+1) \ell} q^\ell A(n - (\nu+1) \ell,\ell-1,\omega) \binom{\omega+1}{i} \binom{\nu-1}{i-1},
		\end{align*}
		where $A(n',l',\omega)$ is the number of all words $\ve{x} \in \mathbb{Z}_q^{n'}$ that have zero-runs of length at most $\ell'$ and Hamming weight $\omega$. 
	\end{lemma}
	\begin{proof}
		We consider the $\ell-$step derivative $\phi_\ell(\ve{x}) = (\ve{u}, \ve{v})$. According to Lemma~\ref{lemma:tandem_deletion_sphere}, the size of the tandem deletion sphere is given by $|S_{\tau^\delta, \ell}(\ve{x})| = wt_\mathrm{H}(\sigma_\ell(\ve{v})) + 1$ and we therefore want to find the number of words $\ve{x} \in \mathbb{Z}_q^n$ with $wt_\mathrm{H}(\sigma_\ell(\ve{v})) = i$.
		
		Let $\nu$ be the number of length $\ell$ tandem duplications in $\ve{x}$, i.e. $|\sigma_\ell(\ve{v})|_1 = \nu$. Further let $\mathcal{J}$ denote the support set of $\sigma_\ell(\ve{v})$, i.e. $\mathcal{J} = \{m : \sigma(\ve{v})_m \neq 0\}$, with $|\mathcal{J}| = i$. The number of possibilities to distribute the duplications into $\sigma_\ell(\ve{v})$ for a given support $\mathcal{J}$ is equal to the number of solutions of
		\begin{equation}
		\sum_{j=1}^{i}y_j = \nu, \quad y_j \in \mathbb{N}, \,\,\forall \,\, 1\leq j \leq i.
		\end{equation}
		This number is given by $\binom{\nu-1}{i-1}$ \cite[Lemma 2.2]{Kulkarni13}. Further, let $\omega$ be the Hamming weight of the trunk, i.e. $wt_{\mathrm{H}}(\mu_\ell(\vec{v})) = \omega$ and thus $|\sigma_\ell(\ve{v})| = \omega+1$, which corresponds to the number of unambiguous positions for tandem duplications of length $\ell$. The number of possible support sets $\mathcal{J}$ of $\sigma_\ell(\ve{v})$ with $|\mathcal{J}| = i$ then is $\binom{\omega+1}{i}$. The vector $\mu_\ell(\ve{v})$ can be chosen to be any $q$-ary vector of length $n-(\nu+1)\ell$ that has zero-runs of length at most $\ell-1$ and Hamming weight $\omega$. The number of such vectors is given by $A(n - (\nu+1) \ell,\ell-1,\omega)$. At last, the first $\ell$ symbols $\ve{u}\in \mathbb{Z}_q^\ell$ can be chosen arbitrarily and thus have $q^\ell$ possibilities. \qed
	\end{proof}
	It can be deduced from the results in \cite{Kurmaev11} that for $\omega\geq 2$ the number of all $q$-ary vectors of length $n'$, maximum zero-run length $\ell'$ and weight $\omega$ is given by
	\begin{equation*}
	A(n',\ell',\omega) = (q-1)^\omega \left\{ \begin{array}{ll}  \sum\limits_{p=0}^{\ell'}\sum\limits_{j=0}^{\omega-1}  (-1)^j \big( 
	%	\begin{pmatrix} 	\omega-1\\j\end{pmatrix}
	\binom{\omega-1}{j}
	\binom{n'-p-1-j(\ell'+1)}{\omega-1}- & \multirow{2}{*}{$n'>\ell'$} \\ \binom{n'-p-1-(j+1)(\ell'+1)}{\omega-1}\big), &  \\
	\binom{n'}{\omega}, & n'\leq \ell'
	\end{array} \right. .
	\end{equation*} 
	For $\omega=0$ and $\omega=1$, it is easy to verify that $A(n',\ell',0) = 1$ if $n'\leq \ell'$, $0$ otherwise, and $A(n',\ell',1) = (q-1)\max \{0, 2(\ell '+1)-n'\}$. 
	\subsection{Bound for Palindromic Deletions of Length $\ell=2$}
	\begin{lemma}
		\begin{align*}
		N_{\pi^\delta}(n,2,i) &= |\{\ve{x} \in \mathbb{Z}_2^n : |S_{\pi^\delta, 2}(\ve{x})| = i+1 \}| = \\ &= 2 \sum_{r_1=0}^{n}\sum_{r_{\mathcal{I},2}=0}^{\frac{n}{2}} \sum_{r_{\mathcal{B},2} = 0}^2 \sum_{r_3=0}^{\frac{n}{3}} \binom{2}{r_{\mathcal{B},2}} \binom{r_1+r_3}{r_3} \binom{r_1+r_3+i-r_{\mathcal{I},2}}{i-r_{\mathcal{I},2}} \cdot \\  & \quad \;\; \binom{r_1+r_3+i + r_{\mathcal{B},2}-2}{r_{\mathcal{I},2}} \binom{n-r_1+r_{\mathcal{I},2}-3(r_3+i)-2r_{\mathcal{B},2}-1}{i-r_{\mathcal{I},2}-1} .
		\end{align*}
	\end{lemma}
	\begin{proof}
		By Lemma \ref{lemma:palindromic_deletion_sphere_l_2} we have to find the number of words $\ve{x}\in \mathbb{Z}_2^n$ with $r_{\mathcal{I},2}(\ve{x}) + r_{\geq 4}(\ve{x}) = i$. Let $r_1, r_2, r_3, r_{\geq 4}$ denote the number of runs of length $1$, $2$, $3$ respectively length at least $4$ in $\ve{x}$. Further, let $r_{\mathcal{I},2}$ denote the number of runs of length $2$ in the interior of $\ve{x}$ and $r_{\mathcal{B},2} \in \{0,1,2\}$ the number of runs of length $2$, that are located at the boundaries of $\ve{x}$. Then, $r_2 = r_{\mathcal{I},2} + r_{\mathcal{B},2}$.
		
		We start with counting the number of words with a given run-distribution $r_1$, $r_2$, $r_3$ and $r_{\geq 4}$. To begin with, we insert runs of length $3$ between the runs of length $1$. In total, there are $\binom{r_1+r_3}{r_3}$ possible such arrangements. We then insert $r_{\geq 4}$ runs of length at least $4$ between these runs of length $1$, respectively $3$. There are $\binom{r_1+r_3+r_{\geq 4}}{r_{\geq 4}}$ possibilities to do so. Next, we insert the $r_{\mathcal{I},2}$ runs of length $2$ into a given constellation of runs of length $1$, $3$ and length at least $4$. As those runs cannot be inserted at the beginning or ending of $\ve{x}$, there are henceforth $\binom{r_1+r_2+r_3 + r_{\geq 4}-2}{r_{\mathcal{I},2}}$ possible combinations. As a final assembling step, we append the $r_{\mathcal{B},2}$ runs of length $2$ to the left and right of $\ve{x}$, where we have $\binom{2}{r_{\mathcal{B},2}}$ possibilities of choosing positions for the runs at the boundaries. Finally, using \cite[Lemma 2.2]{Kulkarni13}, there are
		\begin{equation*}
		\binom{n-r_1-2r_2-3r_3-3r_{\geq 4}-1}{r_{\geq 4}-1}
		\end{equation*}
		possibilities to choose the lengths of the runs of length at least $4$, since there are $n-r_1-2r_2-3r_3$ symbols that can be distributed onto these runs. Substituting $r_{\geq 4} = i - r_{\mathcal{I},2}$ and $r_2 = r_{\mathcal{I},2} + r_{\mathcal{B},2}$ and multiplying by $2$, since the first run can either start with $0$ or $1$ yields the statement. \qed
	\end{proof}
	\subsection{Comparison with Burst Insertion Correcting Codes}
	Figure \ref{fig:bounds_vs_burst} shows the lower bounds (LB) on the redundancy for binary codes and different duplication lengths $\ell$. We compare our results with maximum redundancies of single burst insertion correcting codes from \cite{Schoeny17}. To the best of our knowledge, these constructions have the largest codebooks that can correct a single burst insertion. The figure also includes the redundancies from a single tandem duplication correcting construction with cardinality at least
	\[|\mathcal{C}_\mathrm{VT}(n,\ell)| \geq 2^\ell \cdot \sum\limits_{\omega=0}^{n-\ell} \left\lfloor {\frac{\binom{n-\ell}{\omega}}{\omega+2}}\right\rfloor,\]
	which can be shown to exist based on the principle from \cite{Jain16} using Varshamov-Tenengolts codes~\cite{Varshamov65}. Interestingly, there is a significant gap between the redundancies of existing burst insertion constructions, which motivates a specialized code construction that corrects tandem and palindromic duplication errors.
	\begin{figure}[htp]
		\centering
		\input{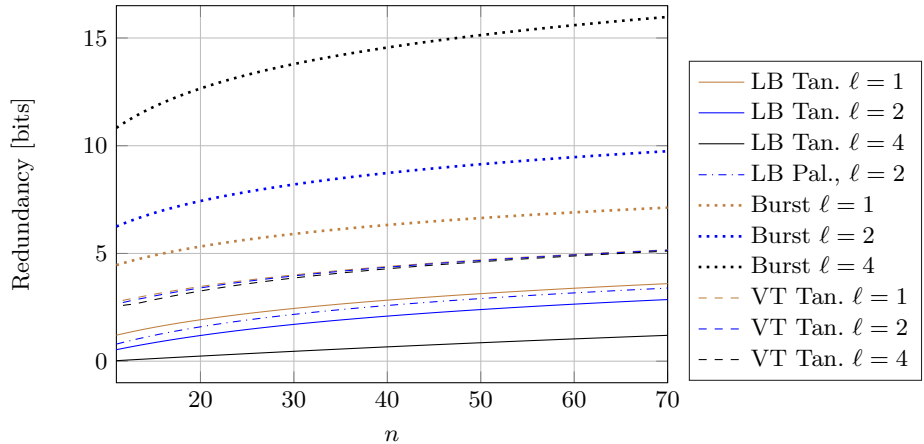}
		\caption{Tandem/palindromic duplication bounds vs. burst insertion redundancies}
		\label{fig:bounds_vs_burst}
	\end{figure}

	\appendix

	\section{Conditions for Equivalence of Palindromic Duplications in one Word} \label{sec:conditions_equivalence_palindromic_duplications}
	In this section we derive conditions that two palindromic duplications, respectively deletions at two different positions $i$ and $i+j$ with $j > 0$ result in the same word $\epsilon_{i,\ell}(\ve{x}) = \epsilon_{i+j,\ell}(\ve{x})$ for $\epsilon \in \{\pi, \pi^\delta\}$. These conditions help to find the sphere sizes $|S_{\epsilon, \ell, 1}(\ve{x})|$, as it has been illustrated in Section \ref{sec:error_spheres_tandem_palindromic_dupliations_deletions}.
	\subsection{Palindromic Duplications}
	For $j < \ell$ the condition $\pi_{i,\ell}(\ve{x}) = \pi_{i+j,\ell}(\ve{x})$ can be expressed as (the left hand side of the equations corresponds to $\pi_{i+j,\ell}(\ve{x})$ and the right hand side to $\pi_{i,\ell}(\ve{x})$)\\
	\vspace{-15pt}
	\begin{subequations}
		\begin{align}
		x_{i+\ell+1+m}  &= x_{i+\ell-m}, \quad m \in \{0, \dots, j-1\}, \label{eq:cond_j_less_l_a} \\
		x_{i+\ell+2j-m}  &= x_{i+\ell-m}, \quad m \in \{j, \dots, \ell-1\},\label{eq:cond_j_less_l_b}  \\
		x_{i+\ell+2j-m}  &= x_{i+1+m}, \quad m \in \{\ell, \dots, \ell+j-1\}. \label{eq:cond_j_less_l_c}
		\end{align}
	\end{subequations}
	For $j \geq \ell$ these conditions are \\
	\vspace{-15pt}
	\begin{subequations}
		\begin{align}
		x_{i+\ell+1+m}  &= x_{i+\ell-m}, \quad m \in \{0, \dots, \ell-1\}, \label{eq:cond_j_geq_l_a} \\
		x_{i+\ell+1+m}  &= x_{i+1+m}, \quad m \in \{\ell, \dots, j-1\}, \label{eq:cond_j_geq_l_b} \\
		x_{i+\ell+2j-m}  &= x_{i+1+m}, \quad m \in \{j, \dots, \ell+j-1\}. \label{eq:cond_j_geq_l_c}
		\end{align}
	\end{subequations}
	\subsection{Palindromic Deletions}
	The conditions $\pi^\delta_{i,\ell}(\ve{x}) = \pi^\delta_{i+j,\ell}(\ve{x})$ for $j>0$ are\\
	\vspace{-15pt}
	\begin{subequations}
		\begin{align}
		x_{i+\ell+1+m}  &= x_{i+\ell-m}, \quad m \in \{0, \dots, \ell-1\}, \label{eq:cond_del_a} \\
		x_{i+\ell+j+1+m}  &= x_{i+\ell+j-m}, \quad m \in \{0, \dots, \ell-1\}, \label{eq:cond_del_b} \\
		x_{i+2\ell+1+m}  &= x_{i+\ell+1+m}, \quad m \in \{0, \dots, j-1\}. \label{eq:cond_del_c}
		\end{align}
	\end{subequations}

	\bibliographystyle{unsrt}
	\bibliography{ref.bib}

\end{document}